\documentclass[12pt,reqno]{amsart}
\usepackage{amsmath,amsthm}
\usepackage[english]{babel}
\usepackage[margin=1in]{geometry}
\usepackage[applemac]{inputenc}
\usepackage{tikz}
\usetikzlibrary{arrows,automata}
\usepackage{amsfonts}
\usepackage{latexsym}
\usepackage{mathtools}
\usepackage{nccmath}
\usepackage{enumerate}
\usepackage{times}
\usepackage{cases}
\usepackage{tikz-qtree}
\usepackage{tikz-cd}
\usepackage[]{graphics}

\allowdisplaybreaks 
\theoremstyle{plain}

\newcommand{\C}{\mathbb C}

\newtheorem{thm}{Theorem}[section]

\newtheorem{cor}[thm]{Corollary}
\newtheorem{defn}[thm]{Definition}

\newtheorem{examp}[thm]{Example}

\renewcommand{\qed}{\hfill $\fbox{}$}

\makeatletter\@addtoreset{equation}{section} \makeatother

\DeclareMathOperator{\per}{Per}

\DeclareMathOperator{\spa}{span}

\newcommand{\lb}{\lambda}
\newcommand{\x}{{\bf x}}
\newcommand{\y}{{\bf y}}
\newcommand{\e}{{\bf e}}
\newcommand{\be}{\begin{equation}}
\newcommand{\ee}{\end{equation}}

\newcommand{\HH}{\mathcal H}

\begin{document}

\title{Relative Reality}
\author{Rongwei Yang\\University at Albany, SUNY}
\thanks{$^*$ Corresponding author}
\address[Rongwei Yang]{Department of Mathematics and Statistics, University at Albany, the State University of New York, Albany, NY 12222, U.S.A.}
\email{ryang@albany.edu}

%\keywords{relative reality}

\maketitle

\begin{abstract}
The ``Hard Problem" of consciousness refers to a long-standing enigma about how qualia emerge from physical processes in the brain. Building on insights from the development of non-Euclidean geometry, this paper seeks to present a structured and logically coherent theory of qualia to address this problem. The proposed theory starts with a definition on what it means for an entity to be non-physical. A postulate about awareness is posed and utilized to rigorously prove that qualia are non-physical and thoughts are qualia. Then the paper introduces a key concept: relative reality, meaning that perceptions of reality are relative to the observer and time. The concept is analyzed through a mathematical model grounded in Hilbert space theory. The model also sheds new light on cognitive science and physics. In particular, the Schr\"{o}dinger equation can be derived easily through this model. Moreover, this model shows that eigenstates also exist for classical energy-conserving systems. Analyses on the G. P. Thomson experiment and the classical harmonic oscillator are made to illustrate this finding. The insight gained sheds new light on the Bohr-Einstein debate concerning the interpretation of quantum mechanics. At last, the paper proposes a postulate about qualia force and demonstrates that it constitutes a fundamental part of absolute reality, much like the four fundamental forces in nature. 
\end{abstract}

\section{Introduction}
Life is the greatest wonder of nature. The ability to perceive and interact with the environment is the defining characteristic of life. Sensation arises from the function of our eyes, ears, nose, tongue, body, and mind. The first five sensory faculties allow us to perceive the external world through sight, sound, smell, taste, and bodily feeling (touch, pain, coldness, etc). The mind observes, analyzes, and integrates the sensory information collected, constructing a cohesive perception of reality. 
Due to this unique capability, the mind is often not regarded as a sensory faculty in Western philosophy. While the physical mechanisms by which the brain generates sensations from neural signals have been extensively studied, long standing mysteries persist about the phenomenological aspects of this process. A critical one, famously articulated by David Chalmers as the ``Hard Problem" of consciousness, concerns how qualia emerge from physical processes in the brain.
Building on insights from the development of non-Euclidean geometry, this article adopts an axiomatic approach to explore this profound problem. Along the way, a mathematical framework is develop, shedding new light on cognitive science and physics.
The paper is structured as follows. 

\tableofcontents

\section{What is Quale?}

The notion of qualia (sg. quale), or ``raw feels", has been intensively scrutinized and debated in neuroscience, cognitive science, and philosophy. Some well-known thought experiments were devised in attempts to grasp its meaning, such as ``What Is It Like to Be a Bat?" \cite[Nagel]{Na74}, ``China Brain" \cite[Block]{Bl80}, ``Mary's Room" \cite[Jackson]{Ja82}, and ``Inverted Spectrum" \cite[Shoemaker]{Sh82}, ``Dancing Qualia" \cite[Chalmers]{Ch96}, etc. As Daniel Dennett famously put: ``qualia" is ``an unfamiliar term for something that could not be more familiar to each of us: the ways things seem to us" \cite{De85}. He identifies four properties of qualia:

{\bf 1}. {\em Ineffable}: they cannot be communicated, or apprehended by any means other than direct experience.

{\bf 2}. {\em Intrinsic}: they are non-relational properties, which do not change depending on the experience's relation to other things.

{\bf 3}. {\em Private}: all interpersonal comparisons of qualia are systematically impossible.

{\bf 4}. {\em Directly or immediately apprehensible by consciousness}: to experience a quale is to know one experiences a quale, and to know all there is to know about that quale. 

\vspace{2mm}

However, the exact meaning of ``qualia" remains elusive in many discussions. This section aims to offer some clarification. 

 \subsection{Non-Physical Entities} Sensation is the raw, unprocessed experience of physical stimuli before it is perceived, interpreted, or organized by the mind. The physical processes leading to sensation have been well-studied, see for example \cite{Go99,Le00,Pa99}. Qualia are the subjective, first-person experience or the ``what it is like" aspect of a sensation.
 It captures the internal, phenomenological quality of experience, and it is therefore inherently subjective. Sensation and quale are obviously closely related, but they are not the same. Sensation is the physical experience of qualia. Qualia are the phenomenological aspect of sensation. For example, when one eats a candy and experiences sweetness in the mouth. The experience part is sensation, while the sweetness is a quale. When we feel headache, the feeling part is sensation, while the pain is a quale. We can measure the intensity of a sensation through the strength of the associated neural signals, or through the body reaction to the sensation. But qualia have no physical entity and cannot be measured directly. Frank Jackson \cite{Ja82} defined qualia as ``...certain features of the bodily sensations especially, but also of certain perceptual experiences, which no amount of purely physical information includes".  In order to better analyze the non-physical nature of qualia, we make the following definition.
   
\begin{defn}\label{Ph} An existence or entity is said to be {\em physical} if it satisfies at least one of the following criteria:
 
 1) A neighborhood of its spatial location can be determined.
 
 2) Its mass, shape, size, or electric charge can be measured.
 
 3) It is one of the four fundamental forces: gravity, electro-magnetic force, strong or weak force, or it interacts with at least one of them. 
 \end{defn}
 
\noindent Consequently, an existence or entity is said to be {\em non-physical} if it bears none of the characteristics described in Definition \ref{Ph}. In the sequel, we shall prove that qualia are non-physical. We take the formation of sound as an example for the analysis. Analyses on sight, smell, taste, bodily feeling, and thought can be done in a similar fashion.
 
\subsection{An Analysis of Sound} 
The creation of sound from airwave stimuli is a process that involves the ear and the brain working together to transform physical vibrations in the air into the experience of hearing. According to the current understanding, its process can be described as follows.

{\bf 1}. {\em Airwaves Entering the Outer Ear}. The pinna (the visible part of the ear) collects and funnels these sound waves into the ear canal, directing them to the eardrum (tympanic membrane).

{\bf 2}. {\em Airwave Causing the Vibration of Eardrum}. These vibrations are transmitted to three tiny bones in the middle ear (the ossicles: malleus, incus, and stapes). The ossicles amplify the vibrations and pass them to the oval window, a membrane leading to the inner ear.

{\bf 3}. {\em Hydraulic Transduction}. Vibrations from the oval window create waves in the fluid inside the cochlea, a spiral-shaped organ in the inner ear. The basilar membrane inside the cochlea vibrates in response to these fluid waves, with different regions responding to different frequencies of sound (high frequencies near the base, low frequencies near the apex).

{\bf 4}. {\em Neural Transduction}. The organ of Corti, located on the basilar membrane, contains hair cells (sensory receptors). The fluid waves cause the hair cells to move, and tiny hair-like structures (stereocilia) on the hair cells bend. This bending opens ion channels, leading to the generation of electrical signals in the hair cells. These signals are passed to the auditory nerve (cochlear nerve).

{\bf 5}. {\em Transmission to the Brain}. The auditory nerve sends the electrical signals to the brainstem, where basic features like timing and localization of sound are processed. The signals are then relayed to the thalamus (a sensory relay station).

{\bf 6}. {\em Processing in the Auditory Cortex}. The signals reach the primary auditory cortex in the temporal lobe of the brain. Here, the brain analyzes features such as pitch, loudness, rhythm, and timbre.

{\bf 7}. {\em Integration and Perception}. The brain integrates the processed signals with input from both ears to create a spatial understanding of sound (e.g., identifying the direction of a sound source). Higher brain areas interpret and give meaning to the sounds, enabling recognition of speech, music, or environmental noises.

{\bf 8}. {\em Creation of the Sensation of Sound}. The mind combines the neural signals and creates the subjective experience of hearing a sound. Context, memory, and attention influence how the sound is perceived and understood (e.g., recognizing a familiar voice or enjoying a melody).

The processes {\bf 1} up to {\bf 7} are physical, while the quale of sound emerges in process {\bf 8}. Sound is the outcome of an interaction between the brain and the air vibration. It is non-physical according to Definition \ref{Ph}. This understanding addresses a folklore philosophical inquisition: when a tree falls in an un-inhabited land, is there a sound? The answer is ``no" because sound is the auditory experience on the part of the listener as his/her neural system interacts with the airwave stimuli. Therefore, if there is no listener, then there is no sound. The preceding analysis paves the ground for the first theorem.

\begin{thm}\label{nonphy}
Qualia are non-physical. 
\end{thm}
\begin{proof} We prove the theorem for sound. It is evident that a normal cat (or any other animal) can hear sound. If sound were physical, as described in Definition \ref{Ph}, then human should be able to measure the sound heard by the cat, namely, measure its location, size, shape, charge, or its interactions with the four fundamental forces. Since none of these measurements can be made, the sound heard by a cat is non-physical. The proof can be generalized to work for qualia of sight, smell, taste, bodily feeling, and thought because they bear the same characteristics as that of sound. In summary, although the formation of each quale has its distinctive physical process, the qualia themselves are all non-physical.\qed
 \end{proof}
Indeed, Theorem \ref{nonphy} is well-understood and commonly accepted. The effort here is to place it on a logical foundation. Since non-physical entities do not have parts, the following property is immediate.
 \begin{cor}
 A quale is inseparable.
 \end{cor}
\noindent For example, when our eyes catch the sight of a house, the quale of this sight is inseparable: even thought it reflects many parts of the house-doors, windows, siding, and roof, etc., the sight is not a union of qualia for each piece.  
Likewise, when we hear a orchestral piece of music, although the sound at any moment is created by an ensemble of violins, cellos, or trumpets, etc., the sound is inseparable. Furthermore, when we see a towering mountain, the sight is not big, when we see a small ant, the sight is not small; when we see a red apple, the sight is not red, when we see a yellow lemon, the sight is not yellow. In short, qualia are devoid of any physical attributes, much like a reflection in a mirror.

\subsection{A Mathematical Analogy} 
Hilbert space, which is a complete vector space equipped with an inner product, plays a fundamental role in many areas of mathematics and physics. In particular, it provides a mathematical framework for quantum mechanics. This subsection aims to draw an analogy between sensation and inner product. For the reader's convenience, we provide some basic definitions.

Let $\HH$ be a complex vector space, finite or infinite dimensional. An {\em inner product} $\langle {\bf \cdot}, {\bf \cdot}\rangle$ is a scalar-valued function defined on the Cartesian product $\HH\times \HH$ that satisfies the following properties. As a convention, we let $0$ stand for the scalar zero as well as the zero vector.
\vspace{2mm}

a) $\langle a{\bf x}+{\bf y}, {\bf z}\rangle=a\langle {\bf x}, {\bf z}\rangle+ \langle {\bf y}, {\bf z}\rangle,\ \ \x, \y, {\bf z}\in \HH, a\in \C$,

b) $\langle {\bf x}, {\bf y}\rangle=\overline{\langle {\bf y}, {\bf x}\rangle}$,

c) $\|{\bf x}\|^2:=\langle {\bf x}, {\bf x}\rangle\geq 0$ with the equality holds if and only if ${\bf x}=0$.
\vspace{2mm}

\noindent Here, $\C$ stands for the set of complex numbers, and $\overline{a}$ stands for the complex conjugate of $a$. 
The value $\|{\bf x}\|$ is the {\em norm} or the length of ${\bf x}$, and $\|{\bf x}-{\bf y}\|$ is the distance between ${\bf x}$ and ${\bf y}$. A notable fact is that the inner product of two vectors is a scalar, not a vector. {\em Cauchy-Schwartz} inequality states that
\[|\langle {\bf x}, {\bf y}\rangle|\leq \|{\bf x}\|\|{\bf y}\|,\ \ \ {\bf x}, {\bf y}\in \HH.\]Two vectors  ${\bf x}$ and ${\bf y}$ are said to be {\em orthogonal} if $\langle {\bf x}, {\bf y}\rangle=0$. The space $\HH$ is said to be {\em complete} if every convergent sequence of vectors in $\HH$ has a limit in $\HH$. In the quantum mechanics terminology, a vector ${\bf x}$ in $\HH$ is called a {\em state}. In this paper, if $\x$ represents a physical state, then $\|{\bf x}\|^2$ is the {\em energy} of $\x$. For more details about Hilbert space, we refer the reader to \cite{Do98,Ha57}. The following is an important example.

\begin{examp}
An {\em inner product} in the vector space ${\C}^n=\{{\bf x}=(x_1, ..., x_n)\mid x_j\in {\C},\ 1\leq j\leq n\}$ is defined by\[\langle {\bf x}, {\bf y}\rangle=x_1\overline{y_1}+\cdots +x_n\overline{y_n}.\] A special case is when all the components of ${\bf x}$ and ${\bf y}$ are real numbers, in which it can be shown that 
\begin{equation}\label{CS}
\langle {\bf x}, {\bf y}\rangle=\|{\bf x}\|\|{\bf y}\|\cos\theta,
\end{equation}
where $\theta$ is the angle between ${\bf x}$ and ${\bf y}$.
\end{examp}

Sensation, as an outcome of the interplay between a neural system and stimuli, is analogous to the inner product in the following sense. 

{\bf 1}. {\em Stimuli as Vector ${\bf x}$}. External stimuli, such as light waves, sound waves, chemical signals, or a combination of them, can be thought of as one vector. This vector represents the external input.

{\bf 2}. {\em Neural System as Vector ${\bf y}$}. The sensory faculties, including individual's prior experiences, knowledge, and expectations, forms the first vector. This vector represents the mental state of an observer which varies with time, individuals, and across different species. 

{\bf 3}. {\em Sensation Represented by Inner Products}. A sensation corresponds to the outcome of the interaction of these two vectors-much like the inner product, which measures the extent to which ${\bf x}$ aligns or interacts with ${\bf y}$. The inner product results in a scalar value that encapsulates an aspect of the interaction, akin to how a sensation encapsulates the subjective experience of a sensory input.

If we regard the Hilbert space ${\HH}$ as a vector space of physical states, then the fact that $\langle {\bf x}, {\bf y}\rangle$ is a scalar, which does not belong to $\HH$, reflects the non-physical aspect of sensation (Theorem \ref{nonphy}). Furthermore, if ${\bf x}$ and ${\bf y}$ are orthogonal, then the stimulus ${\bf x}$ does not generate a sensation on the part of $\y$. Note that every complex number $z\in \C$ has a polar decomposition $z=|z|e^{i\alpha}$, where $\alpha\in [0, 2\pi)$. We summarize the observation above and make the following working definition for this article.

\begin{defn} Given any stimulus ${\bf x}$ and sensory faculty ${\bf y}$, the associated {\em sensation} on the part of $\y$ is defined by $\langle {\bf x}, {\bf y}\rangle=|\langle \x, \y\rangle |e^{i\alpha}$, where $|\langle \x, \y\rangle |$ represents the intensity and $e^{i\alpha}$ represents the quale of the sensation.
\end{defn}

\subsection{Applications to Psychophysics} For the reader who is not familiar with this field, psychophysics is the scientific study of the relationship between physical stimuli and the sensations and perceptions they produce in the human mind. It seeks to understand how physical properties of the environment, such as light and sound wave, are detected, processed, and interpreted by sensory systems. M\"{u}ller's law of specific nerve energies \cite{Mu}, Weber-Fechner law \cite{Fe60}, and Stevens' power law \cite{St57} are fundamental in this area of science. 

To proceed with the discussion on the mathematical model, we recall that the Cauchy-Schwartz inequality (\ref{CS}) relates the inner product $\langle {\bf x}, {\bf y}\rangle$ to the energies of the physical states ${\bf x}$ and ${\bf y}$, offering a systematic quantitative relationship between physical stimuli and the sensations they produce. In equation (\ref{CS}) with both vectors ${\bf x}$ and ${\bf y}$ being real, if we assume $\|{\bf y}\|=1$, then $\langle {\bf x}, {\bf y}\rangle=\|{\bf x}\|\cos\theta$. Therefore, a measurement of sensation $\langle {\bf x}, {\bf y}\rangle$ is a function of $\|{\bf x}\|$. Furthermore, since each sensory faculty has limited range of capability (e.g. human can only see a narrow frequency spectrum of light, and a light too bright can impair the eyes), the domain and range of the function must both be bounded with bounds depending on the physical property of ${\bf x}$. 
\begin{cor}\label{senf}
A measurement of sensation $\langle {\bf x}, {\bf y}\rangle$ is a bounded function of $\|{\bf x}\|$ defined on a bounded interval $a\leq \|\x\|\leq b$.
\end{cor}

The following examples highlight potential applications of this model to psychophysics.

\begin{examp}
The Law of Specific Nerve Energies posits that each sensory faculty, no matter how it is stimulated, produces a specific type of sensation. The functions of eyes, ears, nose, tongue, body, and mind are mutually exclusive (e.g., we cannot see a headache, hear a color, taste a sound, or smell a thought). If we regard the six sensory faculties as distinct physical states in $\HH$, then the law can be described by the orthogonality of these vectors. Interestingly, this description can shed light on why it is impossible to convey a sensation to someone who has never experienced it before: understanding is a function of the mind vector, which is orthogonal to the vectors representing the other five sensory faculties.
\end{examp}

\begin{examp}
Suppose our eye $\y$ catches a photon $\x$ with frequency $\nu$. The energy quanta $E=\|\x\|^2$ of the photon is given by the Max Planck formula $E=h\nu$, where $h$ is the Planck constant. The range of light frequencies that are visible to human is roughly $4.0 \times 10^{14} Hz$ to $7.9 \times 10^{14} Hz$. If we take $10^{14} Hz$ as one unit, then a measurement of color (our sight sensation) can be described as a bounded function defined on $[4, 7.9]$.
\end{examp}

\begin{examp}
The Weber-Fechner law \cite{Fe60} and Stevens' power law \cite{St57} describe empirical relationship in psychophysics, relating the energy $E=\|\x\|^2$ of a stimulus $\y$ to the perceived magnitude of a sensation. They are functions $\psi (E)=k\log (E/E_0)$, and resp. $\phi (E)=k E^{\alpha}$, where $k, E_0$, and $\alpha$ are positive constants dependent on the sensory faculties. Although they are not bounded functions as stipulated by Corollary \ref{senf}, they simulate the psychophysics relation well in certain bounded intervals $a\leq E\leq b$.
\end{examp}

The three examples above demonstrate that the Hilbert space model developed in Section 2.3 aligns with current theories in psychophysics. However, in reality, sensations are often too intricate to be fully captured by complex numbers alone. This model awaits further development to accurately reflect the complexity of sensory experiences.

\section{The Postulate of Awareness}

When the sweetness taste of a candy or the pain of headache is cognized by one's mind, we often say that one is aware of the sweetness or the pain. This example suggests that qualia are the subjects of awareness, and awareness is the act of catching and experiencing qualia. In view of this close connection, a study of qualia is incomplete without an analysis on awareness. For the benefit of subsequent study, it is necessary that we make a clear working definition of awareness. The concept of awareness has been defined in various ways across disciplines such as philosophy, psychology, neuroscience, and artificial intelligence. Here, we list a few recent ones as a reference.

{\bf 1}. Block \cite{Bl95}: Awareness is the ability to directly perceive and consciously experience the world or oneself. It encompasses both external stimuli and internal mental states.

{\bf 2}. Posner and Rothbart \cite{PR98}: Awareness involves the selective focus on certain aspects of the environment or mental states, often described as attention and monitoring.

{\bf 3}. Baars \cite{Ba88}: Awareness is the state in which information becomes globally available for processing across the brain.

{\bf 4}. Shiffrin and Schneider \cite{SS77}: Awareness is the conscious acknowledgment of stimuli or mental processes, distinguishing between implicit (unconscious) and explicit (conscious) cognition.

{\bf 5}. Tononi \cite{To04}: Awareness arises when information is integrated across multiple systems to form a coherent perception.

Although qualia are not explicitly mentioned, being basic elements of the phenomenal aspects of consciousness, they are implicitly used in all the definitions above. This observation leads us to a more foundational definition of awareness in terms of qualia.

\begin{defn}\label{aware}
Awareness is the mind's activity to experience, register, and prioritize qualia.
\end{defn}

Building on the definition, this section aims to propose a postulate about awareness and demonstrate that thoughts are a form of qualia. These findings are instrumental in developing a mathematical framework for cognition. 

\subsection{Awareness and Qualia}

The following analysis aims to further illuminate the relationship between awareness and qualia. 

{\bf 1}. {\em Qualia exist without being aware of them}. For instance, a faint sound in the background may produce auditory qualia (e.g., pitch, tone) even if one is not consciously aware of hearing it. Also, while engrossed in reading, one may not notice the feeling of one's back against the chair, but the tactile qualia of pressure and texture still exist in the background. This indicates that neural processes can register sensory inputs that do not reach the threshold for conscious awareness but still produce qualia at a subconscious level. Hence, depending on one's mental state, awareness tends to prioritize certain experiences while ``dimming" others, though the unattended qualia remain latent.

{\bf 2}. {\em Awareness is independent of specific qualia}. Awareness operates as a selective process that can highlight different qualia in sequence, emphasizing that it is distinct from the content (qualia) it attends to. For example, if a bee stings on a toe of the person engrossed in reading, then his/her awareness moves swiftly from the sight of words to the pain caused by the sting.
During mindfulness practices, awareness can move from focusing on breath to sensations in the hands, feet, or other parts of the body. Furthermore, in deep states of meditation, awareness may persist even when qualia are minimized or absent. Studies show that thoughts are like streams that flow constantly in the subconsciousness mind, while awareness agitates and makes ``most-concerned" thoughts on-stage, for instance, see \cite{Ba97, Fr00, Ka11, Ga18, Th07}. Thus, awareness can be seen as a dynamic higher-order process capable of ``shining a light on" any part of the body or mind. 
At any moment, our body is experiencing numerous sensations, from head to toes and from the skin to the organs. However, our awareness can catch only a few of them. Sensations are abundant even during deep sleep or coma, although hardly any of them is registered by the mind. 

Definition \ref{aware} reveals a fundamental phenomenon. We pose it as a postulate.

\vspace{2mm}
 
\noindent {\bf The Postulate of Awareness}: The subjects of awareness are exclusively qualia.

\vspace{2mm}

\subsection{Some Hidden Aspects of Awareness}

The following discussions explore some additional aspects of awareness that are seldom noted in philosophical discourses yet remain consistent with the postulate presented above.

{\bf 1}. {\em Biological Basis and Energy Costs}. Awareness, being a continuous form of brain activity, requires a substantial amount of bio-energy. This demand is a primary reason that human and animals need sleep. The human brain accounts for about 20\% of the body's total energy usage, even though it constitutes only about 2\% of body weight. Awareness arises from the coordination of distributed neural activity across the brain. The Global Workspace Theory (GWT) \cite{Ba88} suggests that awareness involves the widespread sharing of information among different brain areas, especially through the prefrontal cortex and parietal lobes.
Sensory inputs are processed in primary sensory cortices, while the integration and conscious awareness of these inputs occur in higher-order regions like the parietal cortex \cite{La06}.

Sleep is often viewed as an adaptation to reduce bio-energy expenditure and to allow for essential maintenance processes.
Slow-wave sleep (SWS), in particular, is characterized by reduced neural activity and metabolic rates, conserving energy.
During sleep, the brain clears metabolic waste products (e.g., beta-amyloid), which accumulate during wakefulness. Energy is redirected from active awareness to processes like strengthening synaptic connections and pruning unnecessary ones. While sleep leaves animals vulnerable, the energy conservation and restorative benefits appear to outweigh these risks, as evolution has favored its persistence across virtually all species. 

 {\bf 2}. {\em Time Perception}. Awareness enables us to integrate sensory inputs, memories, and cognitive processes into a coherent experience. Since qualia are processed one after another by the mind, this sequential activity contributes to the perception of time. For example, the awareness of a song involves integrating successive notes into a melody, giving rise to a sense of temporal flow. Without awareness, individual events might be perceived as disjointed. Awareness not only integrates sensory data but also connects past, present, and anticipated future events into a timeline. This process leads to the subjective experience of time as something that ``flows" or ``passes."

From a biological aspect, the perception of time involves specific brain regions that overlap with those enabling awareness \cite{BM11,Ea08,GIM18}. Neural activity, such as theta waves and gamma oscillations, is thought to act as a biological clock that supports the experience of temporal intervals. Awareness depends on these oscillations to integrate information over time, making them central to both time perception and consciousness. Brain areas like the prefrontal cortex and posterior cingulate cortex are active during self-referential thought and temporal awareness. These functions align with the processes required for sustained awareness. The following simple examples illustrate the correlation between awareness and the perception of time:

a) During deep sleep or meditation, when awareness diminishes, time appears to pass quickly. 

b) When awareness is intensely engaged, e.g., during a car accident or moments of danger, time may seem to slow down because the brain processes more details in a shorter interval.

c) Substances that alter awareness, such as psilocybin, can distort the perception of time \cite{SB18,SB22}. This indicates that changes in the structure of awareness directly affect how time is experienced.
 
 {\bf 3}. {\em Awareness Is Local}. Since qualia are inherently tied to specific sensory or cognitive inputs, awareness is generally localized because it is constrained by the spatial and temporal properties of these qualia. Neural mechanisms like the thalamus and prefrontal cortex help filter incoming sensory information, determining which inputs reach conscious awareness. This selective process ensures that awareness focuses on the most relevant stimuli (e.g., a bee sting in your toe) while ignoring less critical inputs. Moreover, cognitive science suggests that awareness has a limited capacity, processing about 4-7 chunks of information at a time \cite{Ba86,Br58,Mi56}. This constraint also explains why awareness cannot encompass the vast array of bodily qualia simultaneously. 

The localization of awareness has practical implications. Techniques like mindfulness-based stress reduction use the localization of awareness to manage pain by shifting attention away from painful areas to less intense sensations. The localized feature of awareness also help doctors to understand and treat disorders like neglect syndrome, in which case individuals ignore one side of their body or environment.

\subsection{Thoughts Are Qualia}

While thoughts are widely accepted as being non-physical, the question whether they are qualia is still in debate. Thomas Nagel \cite{Na74} argues that subjective experiences (including thoughts) are intrinsic to consciousness and cannot be reduced to objective descriptions. David Chalmers \cite{Ch96} discusses how all conscious states, including propositional attitudes like thoughts, have phenomenal aspects. On the other hand, unlike sensory qualia, thoughts are often abstract and propositional, dealing with ideas, logic, and relationships rather than raw, first-person experiences, for example the thought ``97 is a prime". The propositional nature of thoughts seems to suggest they are distinct from the purely phenomenal character of qualia. Jerry Fodor \cite{Fo75} argues that thoughts function more like syntactic structures than qualitative experiences. Furthermore,
thoughts often have intentionality, or ``aboutness" (e.g., thinking about a loved person). Qualia, in contrast, seem to lack this referential quality. Franz Brentano \cite{Br74} introduces the concept of intentionality as distinct from qualia. A hybrid model is proposed in Evan Thompson \cite{Th07}, which discusses how cognition and phenomenology interrelate. 

However, in the framework of this article, the following finding is a logical outcome.

\begin{thm}\label{thought}
Thoughts are qualia.
\end{thm}

\begin{proof}
Evidently, we are aware of our thoughts. Since thoughts are the subjects of awareness, the Postulate of Awareness concludes that they must be qualia.\qed
\end{proof}
Indeed, memory pieces of qualia are the building blocks of thoughts. Thoughts arise as a string of sensory qualia. Even during abstract thinking, such as ``97 is a prime", the inner speech is based on memories of auditory qualia. Abstract concepts (e.g., ``justice" or ``freedom") are also tied to sensory or emotional qualia associated with personal or cultural experiences of these ideas. Moreover, awareness plays an important role to make the thinking process integrated and coherent. For example, one thinks more efficiently when there is sharp concentration on a subject. On the other hand, when one is drowsy or close to fall asleep, one's thoughts often become discrete and meaningless.

In light of Theorems \ref{nonphy} and \ref{thought}, we infer that thoughts are non-physical, which is consistent with our intuitive understanding. 
 
\section{Relative Reality}

``{\em A part of the universe becomes a part of the individual.}" 

\hspace{10.5cm} Davida Teller, \cite{TP} Ch. 6, p. 117

\vspace{2mm}

Although the formation of qualia requires the presence of both observers and physical stimuli, the qualia are not intrinsic to either observers or the physical stimuli. Theorem \ref{nonphy} indicates the non-physical nature of qualia. Studies in cognitive science, for instance \cite{GW63}, reveal that perception is not passive reception but active interpretation. The brain often ``fills in" gaps, filters information, and constructs coherent narratives from raw sensory data. Furthermore, sensations often interact in complex ways (e.g., taste is influenced by smell). This integration further highlights the constructed, non-physical nature of experience, bringing challenges to our understanding of reality and the role of perception in defining existence. This section elaborates on this profound philosophical and scientific issue. 

\subsection{The Reality as We Perceive it}

The concept of qualia challenges the classical view of an objective reality that exists independently of perception. If qualia, which form the basis of how we experience and interpret the world, are neither physical nor intrinsic to the observer or the stimuli, then the perceived reality is fundamentally shaped by the physical structure of the observers. Consequently, ``reality" is a subjective construct relative to the set of observers rather than a fixed, universal truth. This concept of {\em relative reality} aligns with insights from quantum mechanics, where the act of observing or measuring a microsystem influences its state and co-create the phenomena that we perceive as being ``real." It also partially resonates with some religious traditions, such as Mahayana Buddhism. In Buddhist philosophy, sensations are seen as conditioned phenomena arising from interactions between the observer's consciousness and the object. They are impermanent and devoid of intrinsic existence, pointing to the illusory nature of perceived reality. In classical idealist philosophies (e.g., those of Plato, Berkeley, Kant, Hegel) \cite{FG22}, reality is also seen as a construct of the mind. They believe that physical world has no independent existence apart from the sensations and ideas perceived by observers. Phenomenology, as developed by Husserl and others, explores the structure of experience. It emphasizes that our experience of phenomena is mediated through perception, which shapes our view of reality. Sensations, then, are central to the lived experience of existence. We consider the following examples.

{\bf 1}. {\em The Impact of Vision}. For humans, vision is the primary sense for gathering and interpreting information about the environment. This sensory dominance profoundly shapes how we conceptualize and describe reality. Concepts such as color, dimension, distance, shape, and size are deeply rooted in our visual perception, and their roles permeate almost every aspect of art, science, and mathematics. However, human vision is limited to a narrow band of the electromagnetic spectrum (visible light). Infrared, ultraviolet, and other forms of radiation are invisible to us but accessible to other species, meaning our perception of reality, as well as our description of it, are inherently constrained.

{\bf 2}. {\em Perception in Species Without Vision}. Different species perceive the same physical environment in vastly different ways due to variations in sensory capabilities. Species lacking vision rely on other senses to interact with their environment, resulting in profoundly different experiences of reality. For example, bats and dolphins emit sound waves and interpret the returning echoes to navigate, hunt, and understand their surroundings. Thus, their reality is shaped by auditory maps rather than visual landscapes. Moles, which live in dark underground environments, rely on touch-sensitive hairs and vibrations to perceive their surroundings. Their reality is defined by textures, pressures, and motion rather than light and color.
Sharks and some aquatic species detect electric fields emitted by other organisms. This electro-sensory perception enables them to locate prey even in complete darkness or murky water, offering a reality entirely alien to vision-dependent beings.

 Each organism inhabits a species-specific reality shaped by its sensory capabilities and ecological niche. What is real to one species may be irrelevant or imperceptible to another. Making attempts to explore and understand the sensory worlds of other organisms helps us gain a humbling and expansive appreciation for the diversity of existence.
  
\subsection{A Mathematical Model of Cognition}
The mathematical analogy described in Section 2.3 paves way for a rigorous scientific inquisition on the notion of relative reality. A set of vectors $S=\{{\bf e}_1, {\bf e}_2, ...\}$ is said to be an orthonormal basis for the Hilbert space $\HH$ if the following conditions are satisfied:

a) $\langle {\bf e}_i, {\bf e}_j\rangle=\delta_{ij}$, where $\delta_{ij}=0$ when $i\neq j$ and $\delta_{ii}=1$ for all $i,j \geq 1$;

b) Every vector ${\bf x}\in \HH$ can be uniquely expressed as ${\bf x}=x_1{\bf e}_1+x_2{\bf e}_2+\cdots$ for some complex numbers $x_1, x_2, ...$, which are called the coordinates of ${\bf x}$ with respect to the basis $S$. 

\noindent It is not hard to check that $x_j=\langle {\bf x}, {\bf e}_j\rangle$ for each $j$, and the following {\em Parseval's identity} holds:\[\|{\bf x}\|^2=|x_1|^2+|x_2|^2+\cdots.\]

It is worth noting that, although orthonormal basis for $\HH$ is not unique, some bases may be more preferable than others depending on the study at hand. Every finite dimensional subspace $M\subset \HH$, as a Hilbert space itself, has an orthonormal basis. To discuss in more generality, we assume $M$ is a finite dimensional subspace of $\HH$ with a preferred orthonormal basis  $F=\{{\bf e}_j\mid 1\leq j\leq n\}$. In the sequel, such a pair $(M, F)$ (written simply as $M$) shall be called an {\em observer}. The {\em projection operator} $P_M: \HH\to M$ is defined by 
\[P_M{\bf x}=\langle {\bf x}, {\bf e}_1\rangle {\bf e}_1+\cdots +\langle {\bf x}, {\bf e}_n\rangle {\bf e}_n,\ \ \ {\bf x}\in \HH.\]
For convenience, we denote $P_M(x)$ by $\x_M$.
If we regard ${\bf e}_j$ as the $j$th sensory faculty of observer $M$, then the complex number $\langle {\x}, {\e}_j\rangle$ represents the sensation generated by its function in response to the physical stimulus ${\bf x}$. It is worth noting that the projection operator $P_M$ does not depend on the choice of orthonormal basis $F$ for $M$. If $P_M{\bf x}=0$, then we say that the state ${\bf x}$ is not detected by the observer $M$.

\begin{defn} For any ${\bf x}\in \HH$, the {\em perception} of ${\bf x}$ by $M$ is defined as 
\[\per_M({\bf x})=(\langle {\bf x}, {\bf e}_1\rangle, ..., \langle {\bf x}, {\bf e}_n\rangle)\in \C^n.\]
\end{defn}
The perception operator $\per_M$ is a {\em linear map} from $\HH$ to $\C^n$ in the sense that \[\per_M (a{\bf x}+{\bf y})=a\per_M({\bf x})+\per_M({\bf y}),\ \ \ {\bf x}, {\bf y}\in \HH, \ a\in \C.\] It projects a physical state $\x$ in absolute reality to a piece of reality $\per_M(\x)$ relative to the observer $M$. Parseval's identity implies that $\|\x_M\|=\|\per_M ({\bf x})\|$ for every ${\bf x}\in \HH$. The space $\C^n$ above simulates the mind of the observer $M$. If $M$ is a human, then $n=6$, with ${\bf e}_1, ..., {\bf e}_6$ representing our six sensory faculties. It is worth noting here that Theorem \ref{thought} is indispensable. Otherwise, the mind would not be regarded as a sensory faculty and configured into this mathematical model.

\begin{defn}
Given a time interval $a\leq t\leq b$ and a physical state function ${\bf x}: [a, b]\to \HH$. 

a) For any observer $M\subset \HH$, the function $\per_M ({\bf x}(t))$ is called an {\em observation} of ${\bf x}$.

b) The state function $\x$ is called a piece of {\em reality} relative to observer $M$ and the time period $[a, b]$, provided that ${\bf x}_M(t)\neq 0$ for all $t\in [a, b]$.
\end{defn}

Moreover, if $M_1, ..., M_n$ are observers, and $\x_{M_k}(t)$ does not vanish on $[a, b]$ for every $1\leq k\leq n$, then we say ${\bf x}$ is a piece of reality relative to all observers $M_1, ..., M_n$. The following understanding is thus a consequence of Theorems \ref{nonphy}, \ref{thought}, and the discussion above.

\begin{cor}\label{rr}
All perceptions of reality are relative to time and the set of observers.
\end{cor}

The following example offers a meaningful illustration to Corollary \ref{rr}. 

\begin{examp}
We let $H_k$ denote the set of living humans throughout the calendar year $k$. Then the Sun and the Moon are elements in the relative reality of both $H_0$ and $H_{2024}$, whereas cell phones are for $H_{2024}$ but not for $H_0$. Similarly, if Alice is born blind and Bob is born deaf, then music is an element in Alice's reality but not Bob's, while painting is a part of the reality for Bob but not Alice's. Interestingly, neither Alice nor Bob is aware what is missing in their perceptions of reality.
\end{examp}

\section{Logic: Breaking the Boundaries of Relative Reality}

The physical capability of sensory faculties sets boundaries for an observer's perception of reality. This limitation appears to constrain our pursuit for knowledge about nature. For example, we may never know how Beethoven's {\em Pastoral Symphony} sounds to a cat or how candy tastes to a dog. But scientific and technological advancements have helped us significantly extend the capacity of our sensory faculties, thereby expanding our sphere of perceived reality. A powerful tool in this endeavor is logic. While logic appears more universal and self-evident, its principles have in fact also been shaped by our perceived reality, albeit in an abstract and indirect way. One compelling example to illustrate this point is the development of non-Euclidean geometry.

\subsection{A Note on Euclidean Geometry} Ancient civilizations believed that the Earth was flat, with the Sun rising in the East and setting in the West. This perception shaped early worldviews and inspired the development of Euclidean geometry, which assumes a flat, two-dimensional plane as its foundational framework. Five fundamental postulates, as outlined by Euclid (c. 300 BC) in his seminal work {\em Elements}, form the foundation of Euclidean geometry. They are as follows:

\vspace{1mm}
{\em 
a) A straight line can be drawn connecting any two points.

b) A finite straight line can be extended indefinitely in a straight line.

c) A circle can be drawn with any center and any radius.

d) All right angles are equal to one another.

e) The Parallel Postulate: Given a line and a point not on it, there is exactly one line parallel to the given line that passes through the point.}
\vspace{1mm}

\noindent Two lines are said to be parallel if they do not intersect. Note that postulate e) is not the original statement in {\em Elements} but a simplified version of it. Clearly, these postulates align with people's perceived reality in a local scale, although e) is less self-evident than the others. Attempts to prove the Parallel Postulate by the first four postulates have been numerous and extended for many centuries. These attempts undoubtedly all failed, but they finally led to the development of non-Euclidean geometries (elliptic and hyperbolic) in the 19th century \cite{Gr07}, where the Parallel Postulate no longer holds:

a) {\em Elliptic geometry} is founded on the first four postulates together with the assertion that every two ``straight lines" intersect (i.e., there are no parallel lines). This is in fact the geometry of the Earth's surface, where ``straight lines" are big circles with radius equal to that of the Earth. 

b) {\em Hyperbolic geometry} is grounded on the first four postulates together with the assertion that, given a line and a point not on it, there are two or more lines parallel to the given line and pass through the point. Figure 1 below shows an example of hyperbolic geometry. Assume that an observer T is standing at the center of the disc and cannot move (picturing a lone intelligent tree standing on a flat prairie). Due to limited range of its perception, the disc represents the whole world (the prairie) for T. Both lines $\ell_1$ and $\ell_2$ pass through the point $p$ and are parallel to the line $\ell$ from the perspective of T. Note that the point $q$ is not in the world of T. 

Indeed, the development of non-Euclidean geometry constitutes one of the most magnificent pieces of human history, filled with defeat, persistence, and triumph. We can derive some profound insights from it:
 
\begin{figure}[t]
\begin{center}
\scalebox{0.26}[0.26]{\includegraphics{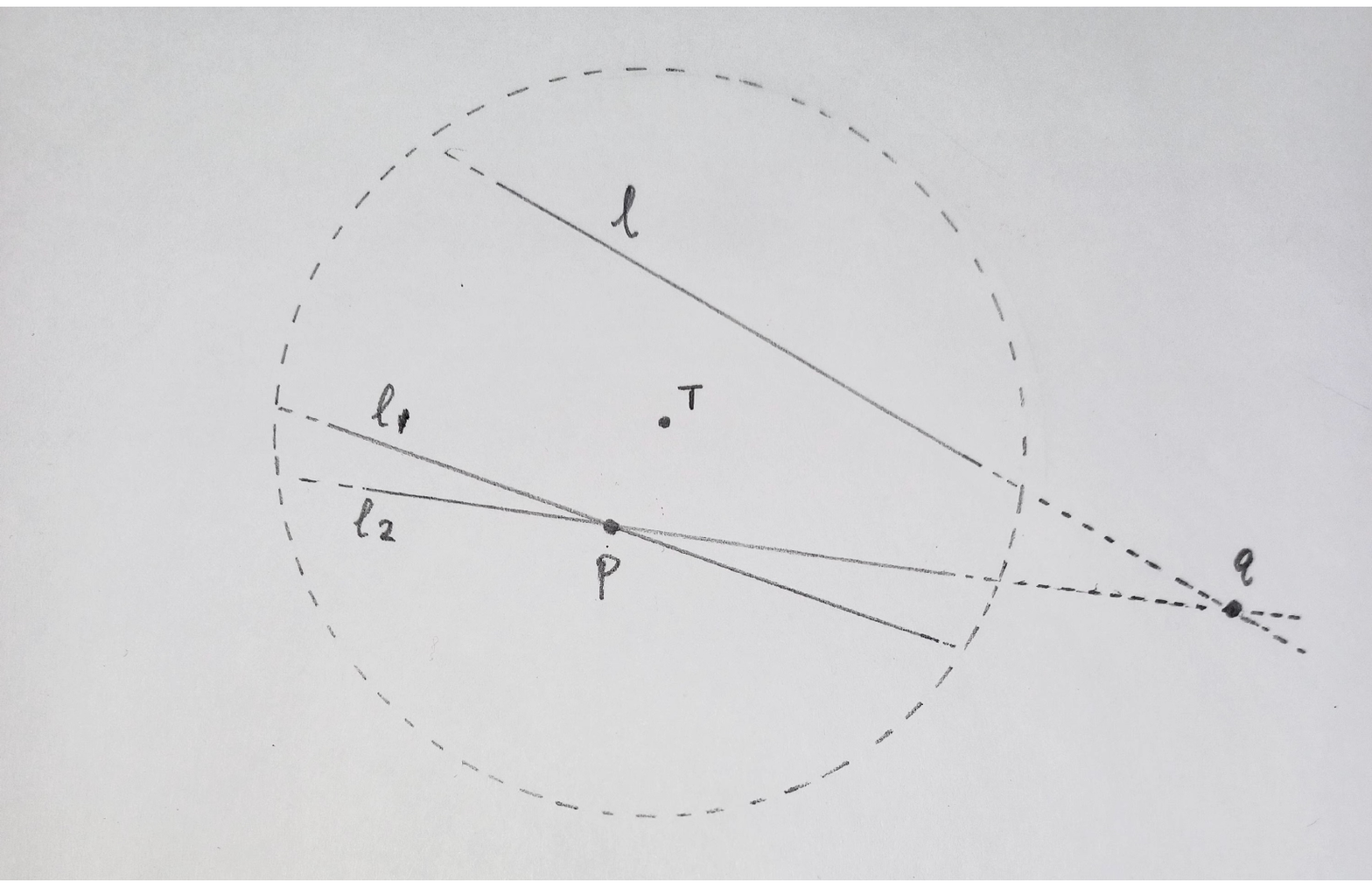}}
\end{center}
\caption{Geometry From the Perspective of a Tree}
\end{figure}

{\bf 1}. {\em Observer-Centric Reality}. Geometric concepts vary depending on the observer. Euclidean geometry describes the local properties of the Earth's surface, approximating it by a tangent plane to the sphere. When technology (e.g., telescopes, satellites) increased human's range of perception, it became evident that our actual geometry is elliptic. However, for an intelligent tree, the geometry might be hyperbolic.

{\bf 2}. {\em The Slow Evolution of Perspectives}. It took humanity over two millennia to challenge and transcend the traditional geometric views shaped by relative reality, eventually arriving at non-Euclidean geometries. This slow progression highlights the deep-seated influence of initial perceptions and the inertia of established paradigms.

{\bf 3}. {\em Logical Consistency in Relative Reality}. Euclidean geometry, despite being founded on relative reality-influenced by an illusory perception of flat Earth-is logically consistent. This fact seems to validate the belief that all phenomena are rational, as  expressed in Hegel's famous dictum, ``What is rational is actual, and what is actual is rational."

{\bf 4}. {\em The Universality of Logic}. Non-Euclidean geometry underscores the enduring validity of logic across different relative realities. The shift from Euclidean to non-Euclidean geometries demonstrates that logical reasoning can bridge diverse conceptual frameworks based on relative realities, whether describing flat planes or curved spaces, celestial bodies or minuscule particles.

Therefore, logical reasoning, often expressed in mathematical forms, is capable of transcending the constraints of relative reality to uncover universal natural laws. The discoveries of Pluto in 1930 and the Higgs boson in 2012 serve as two illuminating examples of this sort. Pluto's discovery was not the result of direct observation but rather the outcome of meticulous mathematical reasoning. Astronomers noticed perturbations in the orbits of Uranus and Neptune, which could not be fully explained by known gravitational forces. Using mathematical models of celestial mechanics, they hypothesized the existence of another massive object-later confirmed to be Pluto. The Higgs boson, often called the ``God particle," was discovered at CERN's Large Hadron Collider after decades of theoretical work. The existence of this particle was first proposed in the 1960s through mathematical models in quantum field theory to explain how particles acquire mass via the Higgs field. This discovery validated the Standard Model of particle physics, a mathematical framework that describes the fundamental forces and particles of the universe.

\subsection{The Origin of Logic Principles}
 
The preceding discussion on geometry invites a deeper exploration of logical principles, particularly to examine the root of their universality. A profound question remains to be addressed: are these principles inherent in the fabric of reality ({\em realism}), or are they constructs shaped by human cognition and perception ({\em constructivism})? Among the most fundamental are the following principles:

\vspace{2mm}
a) {\em The Principle of Identity}: A thing is what it is (e.g., $A$ is $A$).

b) {\em The Principle of Non-Contradiction}: Contradictory statements cannot both be true at the same time (e.g., $A$ and not $A$ cannot both be true).

c) {\em The Principle of Excluded Middle}: For any proposition, either it is true, or its negation is true (e.g., $A$ or not $A$ must be true).
\vspace{2mm}

\noindent Human reasoning, grounded in language, suggests that languages have played a pivotal role in shaping and formalizing the logical principles. Notably, the set of nouns in a language establishes a one-to-one correspondence between objects in the natural environment and human cognition of them. This correspondence effectively transfers basic natural laws into the mental domain. For instance, the nouns ``apple" and ``strawberry" denote two distinct entities, each possessing a unique identity that remains consistent throughout its existence. This understanding underpins the Principle of Identity. Furthermore, the observation that an apple cannot simultaneously be a strawberry reflects the Principle of Non-Contradiction. Finally, the realization that any given fruit is either an apple or not an apple encapsulates the Principle of Excluded Middle. These logic principles mirror universal properties of natural phenomena across different relative reality contexts. For instance, the afore-mentioned properties of apple and strawberry remain valid from the perspective of other animals. If cats were to develop a complex language system like ours, it is plausible that they will also abide by these logic rules. 

The use of nouns has been a fundamental aspect of human communication since the dawn of civilization. Over time, the rules governing their use evolved into basic logical principles. To address the question posed at the beginning of this subsection, although languages arise from our perception of relative reality, the logical principles reflects universal properties of nature that transcend human experience. With the advent of mathematics and philosophy, logic became an abstract system, detached from particular relative reality contexts, enabling us to explore and reason beyond the boundaries of our immediate perception.  This understanding is a cornerstone of this article.

\section{An Application to Physics}

Hilbert space and projection operators are foundational elements for the quantum theory. By making a connection with qualia and perception, the model developed in Sect. 4.2 presents a philosophical explanation for their prominence. This section describes a technical aspect of this model by exploring an application to physics.

If an experiment has $n$ outcome states $\e_j\in\HH, j=1, ..., n$, then $M:=\spa \{\e_1, ..., \e_n\}$ can be regarded as an observer associated with the experiment. Conversely, an arbitrary observer $M\subset \HH$ with an orthonormal basis $\{\e_1, ..., \e_n\}$ can also be viewed as an experiment with outcome states $\e_1, ..., \e_n$. In this section, the terminologies ``experiment", ``observer", and ``subspace" refer to the same thing, but they are used in different places depending on the context. Recall that $P_M$ denotes the orthogonal projection from $\HH$ onto the subspace $M$. From the perspective of $M$, every state in $\HH$ is a superposition of the vectors $\e_1, ..., \e_n$ (which is correct only when $\x\in M$). If $\x$ and $\y$ are states in $\HH$, then the inner product $\langle \x, \y \rangle$ can be viewed as a piece of {\em information} generated by and shared between $\x$ and $\y$. In the case $P_M\x=0$, i.e., $\langle \x, \e_j \rangle=0$ for each $j$, we can say that $M$ does not detect the existence of $\x$. Furthermore, a {\em physical process} (or {\em state function}) refers to a differentiable function $\x: [0, \infty)\to \HH$, and it is said to {\em conserve energy} if $\|\x(t)\|$ is constant. This framework can sheds some new light on energy conserving physical processes. The discussion here is preliminary. More explorations will be made in another paper.

\subsection{Schr\"{o}dinger's Equation} 

If a physical process $\x(t)$ conserves energy, then it evolves according to Schr\"{o}dinger's equation
\begin{equation}\label{sch}
i \hbar\frac{d\x(t)}{dt}=H\x(t),\ \ \ t\geq 0,
\end{equation}
where $i=\sqrt{-1}$, $\hbar$ is the reduced Planck constant, and $H$, often referred to as the {\em Hamiltonian operator} or {\em energy operator}, is a Hermitian linear operator on $\HH$ (possibly unbounded) that characterizes the evolution. In quantum mechanics, such operator $H$ is also called an {\em observable}. The derivation of this equation is simple and worth seeing. The conservation of energy implies that the evolution is {\em unitary} in the sense that $\x(t)=U(t)\x_0$, where $U(t), t\geq 0$, is a differentiable function with unitary operator values that satisfies the following properties:

a) $U(0)=I$, where $I$ stands for the identity operator,

b) $U(s+t)=U(s)U(t)$,

\noindent Condition a) sets the initial condition $\x(0)=\x_0$, and condition b) means that the evolution from time $t'=0$ to $t'=s+t$ is equal to the evolution from $t'=0$ to $t'=s$ followed by the evolution from $t'=s$ to $t'=s+t$. Then, using a) and b), one has
\begin{align*}
\frac{d\x(t)}{dt}&=\lim_{\Delta t\to 0}\frac{1}{\Delta t}(U(\Delta t+t)\x_0-U(t)\x_0)\\
&=\lim_{\Delta t\to 0}\frac{1}{\Delta t}(U(\Delta t)-I)U(t)\x_0=U'(0)\x(t).
\end{align*}
Differentiating the equation $U(t)U^*(t)=I$, one obtains
\[U'(t)U^*(t)+U(t)(U^*)'(t)=0.\]
Setting $t=0$ and using a), one sees that $U'(0)+(U^*)'(0)=0$, i.e., $U'(0)$ is skew Hermitian. Hence, if we set $H=\frac{-i}{\hbar}U'(0)$, then $H$ is Hermitian, and one arrives at equation (\ref{sch}). It is worth noting that, from this perspective, the appearance of the imaginary number $i$ is in fact a consequence of the conservation of energy. Equation (\ref{sch}) is not hard to solve:
\begin{equation}\label{solution}
\x(t)=e^{\frac{-i}{\hbar}Ht}\x(0), \ \ \ t\geq 0.
\end{equation}
\noindent Since $H$ commutes with $e^{\frac{-i}{\hbar}Ht}$, taking the derivative of $\x(t)$ in (\ref{solution}) we obtain
\[\x'(t)=\frac{-i}{\hbar}e^{\frac{-i}{\hbar}Ht}H\x(0), \ \ \ t\geq 0,\]
which implies $\|\x'(t)\|=\frac{1}{\hbar}\|H\x(0)\|=\|\x'(0)\|$. 

\begin{cor}\label{uncertainty}
If a physical process $\x(t)$ conserves energy, then so does $\x'(t)$.
\end{cor}
\noindent The meaning of this corollary will be addressed later in this section.

If $M\subset \HH$ is an observer, then applying the projection $P_M$ to equation (\ref{sch}), one has 
\begin{align*}
i \hbar\frac{d\x_M(t)}{dt}=P_MH\x(t)=P_MHP_M\x(t)+P_MH(I-P_M)\x(t).
\end{align*}
Denoting $P_MHP_M$ by $H_M$ and using the property $P_M=P_M^2$, one arrives at 
\begin{equation}\label{schM}
i \hbar\frac{d\x_M(t)}{dt}=H_M\x_M(t)+P_MH(I-P_M)\x(t).
\end{equation}
Note that if $\x(t)$ evolves inside $M$, i.e., $\x(t)\in M$ for each $t\geq 0$, then $(I-P_M)\x(t)=0$, and equation (\ref{schM}) reduces to the Schr\"{o}dinger equation (\ref{sch}). There is another scenario for this to happen. We make the following definitions before continuing.

\begin{defn}
An observer $M\subset \HH$ is said to be {\em invariant} for $\x$ if $P_MH=HP_M$.
\end{defn}
\noindent Multiplying the equation $P_MH=HP_M$ by $P_M$ on the right, we obtain $P_MHP_M=HP_M=P_MH$, which implies that $P_MH(I-P_M)=(I-P_M)HP_M=0$. In terms of operator theory, such an $M$ is called an {\em invariant subspace} for $H$. It is well-known that every Hermitian operator on $\HH$ has nontrivial (i.e. $M\neq \{0\}$ or $\HH$) invariant subspaces \cite[Ch X. sect. 4]{Con}. 

\begin{cor}\label{inv}
Assume that a physical state $\x$ evolves unitarily in $\HH$. Then $\x$ has nontrivial invariant observers.
\end{cor}
\noindent  In this case, equation (\ref{schM}) loses the second summand on the right and consequently has the solution $\x_M(t)=e^{\frac{-i}{\hbar}H_Mt}\x_M(0),\ t\geq 0$. Thus, from the perspective of $M$, the evolution of $\x$ conserves energy.

However, in the case $M$ is not invariant for $\x$, the evolution of $\x_M$ may not be unitary (see the next subsection), and hence the conservation of energy may not hold from the perspective of $M$. One method to address this issue is to enlarge $M$ to a bigger subspace that is invariant for $\x$. This is in fact a common practice in physics.

\subsection{Finite Dimensional Observers}

Since an observer $M\subset \HH$ can only measure a finite number of aspects of a physical state $\x$, it makes sense to assume that $M$ has a finite dimension $n$. Then the operator $H_M$ in (\ref{schM}) is a finite rank Hermitian operator acting on $M$. Let $S:=\{\e_1, ..., \e_n\}$ be an orthonormal basis for $M$ and assume $\x_M(t)=x_1(t)\e_1+\cdots +x_n(t)\e_n$, where each $x_j$ is a complex scalar functions describing the evolution of $\x(t)$ from the perspective of observer $\{\e_j\}$. The dimension $n$ can be thought of as the degree of freedom for the state $\x_M$. If $\x_M(t)$ conserves energy, then Corollary \ref{uncertainty} takes the form
\begin{equation}\label{uncer1}
|x'_1(t)|^2+\cdots +|x'_n(t)|^2=\|\x'_M(0)\|^2.
\end{equation} 
%%
\iffalse
We define $\Delta \x=\x_M(t+\Delta t)-\x_M(t)$. When $\Delta t$ is close to $0$, calculus gives $\Delta \x\approx \x'(t)\Delta t$, and equation (\ref{uncer1})) leads to the following result.

\begin{cor}
If $\x_M(t)$ conserves energy, then
$\|\Delta \x\|\approx \|\x'_M(0)\|\Delta t$.
\end{cor}
\noindent Since any measurement on the evolution of $\x_M$ takes some amount of time, this corollary in some way is reminiscent of Heisenberg's uncertainty principle (???).
\fi
%%
In the case the basis $S$ is formed by eigenstates of $H_M$ with corresponding eigenvalue $\lb_j, 1\leq j\leq n$, the matrix representation of $H_M$ with respect to $S$ is diagonal with entries $\lb_j$. If the evolution $\x_M(t)$ conserves energy, then solving equation (\ref{sch}) with $\x$ and $H$ being substituted by $\x_M$ and resp. $H_M$, we obtain
\begin{equation}\label{conservE}
\x_M(t)=e^{-i\lb_1t/\hbar}x_1(0)\e_1+\cdots +e^{-i\lb_nt/\hbar}x_n(0)\e_n,
\end{equation}
showing that every eigenstate component $x_j(t)\e_j$ of $\x_M(t)$ evolves unitarily. 
\begin{cor}\label{conserv}
If $\x_M$ conserves energy, then there is an orthonormal basis $\{\e_1, ..., \e_n\}$ consisting of eigenstates of $H_M$ with respect to which $\x_M$ is of the form (\ref{conservE}).
\end{cor}
\noindent In particular, as noted after Corollary \ref{inv}, if $M$ is an invariant observer, then the evolution of $\x_M(t)$ conserves energy and hence is of the form (\ref{conservE}). However, if we choose a different orthonormal basis for $M$, say $\{\hat{\e}_1, ..., \hat{\e}_n\}$ not consisting of the eigenstates of $H_M$, and write $\x_M(t)=\hat{x}_1(t)\hat{\e}_1+\cdots +\hat{x}_n(t)\hat{\e}_n$, then $|\hat{x}_j(t)|$ may not be constant, and it could even vanish, although the sum remains constant: \[ |\hat{x}_1(t)|^2+\cdots +|\hat{x}_n(t)|^2=\|\x_M(t)\|^2=\|\x_M(0)\|^2,\ \ \ 0\leq t\leq b.\]

When $M$ is not invariant for $\x$, the term $P_MH_M(I-P_M)\x(t)$ in equation (\ref{schM}) may not vanish, and we can write it as $h_1(t)\e_1+\cdots +h_n(t)\e_n$, where $h_j, j=1, ..., n$, are scalar functions. Then, equation ({\ref{schM}) turns into a system of linear equations 
\[i\hbar\frac{dx_j}{dt}=\lb_jx_j+h_j(t),\]which has solutions 
\begin{equation}\label{evolution}
x_j(t)=e^{-i\lb_jt/\hbar}\left(x_j(0)+\frac{-i}{\hbar}\int_0^t h_j(s)e^{i\lb_js}ds\right),\ \ \ j=1, ..., n.
\end{equation}

\subsection{The G. P. Thomson Experiment} Before proceeding with more theoretic discussions, we use an example to illustrate the meaning of the coefficient functions $x_j(t)$ in (\ref{conservE}). If $M=\{\e\}$ is $1$-dimensional and conserves the energy of $\x(t)$, then equation (\ref{conservE}) gives $\x_M(t)=e^{-i\lb t/\hbar}x(0)\e$. Note first that if $\lb=0$, then $\x_M(t)$ is in a constant state. For $\lb\neq 0$, given any constant state $\y\in \HH$, the inner product $\langle \x_M(t), \y\rangle=e^{-i\lb t/\hbar}x(0)\langle \e, \y\rangle$ is a complex number reflecting a piece of information about the interaction of $\x_M$ with $\y$ at time $t$. The set of such numbers \[C_\x:=\{\langle \x_M(t), \y\rangle \in \C\mid t\geq 0\}\] is a circle with center $0$ and radius $|x(0)\langle \e, \y\rangle |$. So from the perspective of the observer $\{\y\}$, the physical process $\x$ is a rotation with frequency $\lb/\hbar$. 

Suppose $\x_M^n, n=1, 2, ...$ is a sequences of identical physical processes with possibly a finite number of different initiate states $\x_M^n(0)$, for instance a ray of electrons. Then $\langle \x_M^n, \y\rangle$ is a sequence of complex numbers that gradually fill up several concentric circles with radii $|x^n(0)\langle \e, \y\rangle |$. 

In the 1927 George P. Thomson experiment \cite{Th27}, a fine beam of homogeneous cathode rays is sent nearly normally through a thin celluloid film about $3\times 10^{-6}$ cm. thick and then received on a photographic plate $10$ cm. away and parallel to the
film. Thomson found that the central spot formed by the undeflected rays is surrounded by some concentric rings whose radii decrease with the increasing energy of the rays. This experiment provided crucial evidence for the wave-like nature of electrons, confirming the de Broglie matter wave hypothesis. 

In the observation above, we let $\y$ represents the photographic plate and $\x_M^n, n=1, 2, ...$ represent the ray of electrons deflected by the celluloid film. Then the electrons have varying initial states $\x_M^n(0)$ caused by the deflection, enabling them to produce concentric circles of varying radii on the plate. This seems consistent with the outcome of the Thomson experiment. However, this illustration is hypothetical because the presumed connection between the initial states $\x_M^n(0)$ and the deflection by the film is not tested. If this hypothesis is indeed confirmed, then we know that the wave nature of the electron is due to the factor $e^{-i\lb t/\hbar}$ in $\x_M$ which is rooted in the conservation of energy. The expression $\x_M(t)=e^{-i\lb t/\hbar}x(0)\e$ then embodies the particle-wave duality.

\iffalse
\begin{examp}
Assume $\dim M=4$ and has an orthonormal basis $\{\e_x, \e_y, \e_z, \e_m\}$. Then an object with mass $m$ and position $(x(t), y(t), z(t))$ at time $t$ can be described by the state function $\x_M(t)=c(p_x(t)\e_x+p_y(t)\e_y+p_z(t)\e_z+m_oc\e_m)$, which corresponds to the four-momentum in relativity theory. %Thus, Newtonian mechanics (in which $m(t)$ is constant) can be phrased in this framework. Moreover, if we stipulate that the evolution of $\x_M(t)$ preserves the Minkowski metric $(ct)^2-x^2(t)-y^2(t)-z^2(t)$, where $c$ is the speed of light in vacuum, then we can derive Einstein's special relativity theory. It is worth noting again that, if $\x_M(t)$ conserves energy, there are eigenstates for $H_M$ such that $\x_M$ can be described by (\ref{conservE}). This suggests that the $(x, y, z)$-coordinate, which is grounded in the relative reality shaped by our vision, is probably not optimal.
\end{examp}
\fi

\subsection{Harmonic Oscillator}
Harmonic oscillator refers to systems where the force acting on an object is negatively proportional to its displacement from equilibrium. Without friction, the sum of the potential energy and the kinetic energy is conserved. Such systems play a pivotal role in both classical and quantum mechanics. In a classical spring-mass system without friction, the energy $E$ is given by 
\[E=\frac{1}{2}k\ell^2+\frac{p^2}{2m},\]
where $k$ is the spring constant, $\ell$ is the displacement from equilibrium, and $p$ is the momentum of the mass $m$. The first term on the right is the potential energy and the second is the kinetic energy. In the case of quantum harmonic oscillator, the Hamiltonian is of the form
\[H=\frac{1}{2}k\hat{\ell}^2+\frac{\hat{p}^2}{2m},\]
where the position operator $\hat{\ell}$ and the momentum operator $\hat{p}$ are unbounded Hermitian operators densely defined on the Hilbert space $L^2(\mathbb{R})$, and the eigenvalues of $H$ are quantized energy levels of the oscillator. Details about quantum harmonic oscillator can be found in any quantum mechanics textbook, for instance \cite{Sh94}. As an illustration of Corollary \ref{conservE}, this subsection determines the eigenstates for the classical harmonic oscillator.

We let $\e_{\ell}$ and $\e_p$ stand for the unit position and momentum states, respectively. Then they are orthogonal and span a two dimensional observer $M$. If the evolution of the system is described by $\x(t)$, then $\x_M(t)=x_{\ell}(t)\e_\ell+x_p(t)\e_p$, where the scalar functions satisfy: \[|x_\ell(t)|^2=\frac{1}{2}k\ell^2(t), \hspace{2cm} |x_p(t)|^2=\frac{p^2(t)}{2m},\]
and both quantities varies with respect to time. For simplicity, we set the initial condition $\ell(0)=0, \ell'(0)=1$, and let $\omega:=\sqrt{\frac{k}{m}}$. Solving the differential equation $\ell''(t)+\omega^2 \ell=0$, we obtain $\ell(t)=\sin (\omega t)/\omega$. Taking into account the relation $p(t)=m\ell'(t)$, we have $x_\ell(t)=\sqrt{\frac{m}{2}}\sin (\omega t)$ and $x_p(t)=\sqrt{\frac{m}{2}}\cos (\omega t)$. 

Interestingly, according to Corollary \ref{conserv}, there exists an orthonormal basis $\{\e_1, \e_2\}$, consisting of eigenstates for $H_M$, such that $\x_M$ is given by the form (\ref{conservE}) for $n=2$. Indeed, some computations can verify that one such basis is given by
\[\e_1=\frac{1}{\sqrt{2}}(\e_\ell-i\e_p), \hspace{2cm} \e_2=\frac{1}{\sqrt{2}}(\e_\ell+i\e_p).\]
If we write $\x_M(t)=x_1(t)\e_1+x_2(t)\e_2$, then it can be verified that
$x_1(t)=\frac{i}{2}\sqrt{m}e^{-i\omega t}$ and $x_2(t)=\frac{-i}{2}\sqrt{m}e^{i\omega t}$. It follows that
$\|\x'_M(t)\|^2=\frac{1}{2}m\omega^2,$ which is a constant scalar multiple of the angular kinetic energy.

Since all closed physical systems, classical or quantum, conserves energy, Corollaries \ref{uncertainty} and \ref{conservE} should have some universal implications. The observations above offer the following insights:

{\bf 1}. Corollary \ref{uncertainty} is manifested in the classical harmonic oscillator example as the conservation of angular kinetic energy. Therefore, it can be viewed as a generalization of the latter. Will we obtain more conservation laws by using the corollary repeatedly?

{\bf 2}. From the perspective of the 1-dimensional observers $\{\e_j\}, j=1, 2$ above, the evolution of $\x_M$ is circular and has angular momentum, which is not the case from the perspective of the observers $\{\e_\ell\}$ and $\{\e_p\}$. This highlights the fact that different observers can have different perceptions about the same physical process.

{\bf 3}. A classical energy-conserving system also has eigenstates, and they provide an efficient alternative basis for the description of physical phenomena. Therefore, all physical aspects of the system can be expressed deterministically as 
superpositions of the eigenstates, with the coefficients $x_j(t)$ describing the energy of the system from the perspective of the observers $\{\e_j\}$. Since Corollary \ref{conservE} applies to all energy-conserving systems, there is no reason to believe that a quantum system should behave differently. This confirms Einstein's belief that ``God does not play dice". However, when a large number of identical particles $\x^n(t), n=1, 2, ...$ are measured by an observer $M$, the measurement somewhat randomly sets the initial conditions $\x^n(0)$, leading to the probabilistic nature of quantum superposition described by the Copenhagen interpretation \cite{Fa19}. 

In summary, although the discussion in this section is preliminary, there are many avenues to explore along these lines. For instance, it would be intriguing to determine the eigenstates of other classical energy-conserving systems and explore their practical applications. It is also tempting to analyze other experiments, such as the double-slit experiment \cite{Fe85} and the Casimir phenomenon \cite{La05}, using this model. Indeed, there is potential to expand this section into a comprehensive physics theory that unifies classical and quantum mechanics while transcending human-centric perceptions of nature.

\section{The Postulate of Qualia Force}

``{\em It must be confessed, moreover, that perception, and that which depends on it, are inexplicable by mechanical causes, that is, by figures and motions. And, supposing that there were a mechanism so constructed as to think, feel and have perception, we might enter it as into a mill. And this granted, we should only find on visiting it, pieces which push one against another, but never anything by which to explain a perception. This must be sought, therefore, in the simple substance, and not in the composite or in the machine."}

 \hspace{9cm} Gottfried Leibniz, Monadology, Sect. 17

\vspace{2mm}

\noindent The ``Hard Problem" of consciousness refers to the long-standing challenge of explaining why and how subjective experiences arise from physical processes in the brain. While neuroscience can address many ``easy problems" of consciousness-such as understanding the physical aspects of perception, attention, and behavior-the hard problem focuses on the existence of qualia and the nature of subjective awareness itself. This problem is at the core of many philosophical debates. The thought experiments mentioned in Sect. 2 illustrates the difficulty of reducing consciousness to physical states. This section aims to further analyze qualia and contribute to this discussion.

\subsection{The Principle of Causal Closure}
The principle in the title states that all physical effects must have sufficient physical causes, meaning that everything that happens in the physical world can only be caused by other physical things within that world, excluding any non-physical influences. In other words, the physical realm is a ``closed system" where only physical events can cause other physical events. All observed phenomena in physics-whether macroscopic or quantum-are consistent with the function of the four fundamental forces, leaving no room for non-physical influences to intervene in a measurable way. Furthermore, from a logical point of view, if a process is entirely defined and constrained by physical principles (e.g. properties of the four fundamental forces), then any outcome that deviates from these principles would contradict the definition of a physical process. Therefore, in light of Theorem \ref{nonphy}, the Principle of Causal Closure provides sufficient justification for us to pose the following postulate:

\vspace{2mm}

\noindent {\bf The Postulate of Qualia Force}. There is a natural ``force", previously unidentified and unrecognized in science, that underpins the creation of qualia.
\vspace{2mm}

\noindent This force, called the {\em qualia force} in this paper, should not be understood as something similar to the four fundamental forces, since it has rather distinct property: it can create qualia. Theorem \ref{nonphy} implies that acknowledging qualia force does not contradict physics, which is grounded on the four fundamental forces. On the other hand, physics also cannot prove the existence of qualia force. This situation is very much similar to the independence of the Parallel Postulate (or alternative postulates for elliptic and resp. parabolic geometry) from the first four postulates in Euclidean geometry.
Therefore, if we adopt the four fundamental forces as natural postulates, then adding the Postulate of Qualia Force will give rise to a logically consistent new theory. Such a theory helps integrate qualia into the causal framework of the physical world without violating the Principle of Causal Closure, and it would offer several significant theoretical advantages, both philosophically and scientifically:

{\bf 1}. {\em Bridging the Explanatory Gap}. Qualia, as subjective phenomena, does influence physical actions (e.g., moving your hand due to pain). Introducing the qualia force maintains causal closure by providing a naturalistic explanation for these interactions, ensuring that qualia are not treated as violations of physical laws.

{\bf 2}. {\em Extending the Framework of Physics}. Physics has historically expanded its scope by incorporating new forces to explain previously unexplained phenomena (e.g., gravity for celestial motion, electromagnetism for light and electricity). The qualia force could similarly extend physical theories to account for the nature of consciousness, incorporating subjective experience into the broader physical framework.

{\bf 3}. {\em Resolving the Mind-Body Problem}. This could lead to the unification of mental and physical phenomena, enhancing our understanding of reality. Recognizing the qualia force could offer a dual-aspect {\em monism}, where qualia and physical states are two facets of the same underlying process governed by natural laws. This avoids the pitfalls of {\em dualism} (which struggles with interaction) and {\em reductionism} (which denies the irreducibility of qualia), reframing debates in metaphysics, philosophy, and epistemology by grounding subjective experience in natural science. 

{\bf 4}. {\em Driving Scientific Inquiry}. Acknowledging a qualia force as a natural phenomenon would encourage empirical research into the relationship between qualia and the physical processes, known as the {\em linking theory} \cite{Te84,TP}. It opens new avenues in life science, neuroscience, psychology, and physics to investigate natural phenomena.

{\bf 5}. {\em Solving the ``Hard Problem" of Consciousness}. The concept of qualia force could offer a principled explanation for the emergence of consciousness, showing how certain physical systems generate subjective experiences.

{\bf 6}. {\em Resolving the mystery in teleonomy}. Biological systems, such as genes, cells, ant colonies, and human body, exhibit an apparent but enigmatic coordination of goal-oriented behaviors \cite{CK23,Go14,MS02}. Given that qualia force interacts uniquely with living systems, it could be considered fundamental to how these systems perceive and engage with their environments, thereby enhancing their adaptability and purpose-driven actions.

{\bf 7}. {\em Providing Insight on the Origin of Life}. Most importantly, since sensation is a fundamental characteristic of life, the recognition of a qualia force offers a promising mechanism to bridge gaps in our understanding of life's origins, evolution, and purpose.  It can provide a framework to study how life emerged and adapted through natural laws, offering a deeper and more holistic perspective on the nature of life itself.

\subsection{Absolute Reality}

Since qualia are dependent on the biological structure of the sensory faculties, as demonstrated in earlier sections, the phenomenal worlds that we take as reality is undoubtedly relative (Corollary 4.3). The notion of {\em absolute reality} refers to the existence of entities and forces that are independent of subjective perception or measurement. It is what ``exists" regardless of whether it is observed, measured, or conceptualized. The four fundamental forces are widely considered part of absolute reality because they govern interactions in the universe consistently, irrespective of observation. The following is thus a logical consequence of the Postulate of Qualia Force.

\begin{thm}\label{abr}
 The qualia force constitutes a part of absolute reality.
 \end{thm}
 
The following three characteristics of qualia force help illuminate Theorem \ref{abr}.

{\bf 1}. {\em Universality Across All Life Forms}. There are a great variety of life forms on the Earth, from minuscule single cells to large mammals, and each of them relies on qualia force to perceive and adapt to the environment. This universality indicates a fundamental aspect of existence, akin to the other four fundamental forces. 

{\bf 2}. {\em Objectivity of Function}. The qualia force operates objectively, producing sensations that reflect external processes without subjective bias in the mechanism itself. For example, sight reflects everything within its range, like a mirror, irrespective of aesthetic or emotional judgments about the observed objects. Sensory systems, guided by the qualia force, function like instruments, without imposing subjective preferences.

{\bf 3}. {\em Timelessness}. If an Earth-like planet were to exist billion years ago or in the future, then similar life forms will evolve under comparable conditions. The timeless operation of the qualia force suggests that life forms are not limited to a specific era or location in the cosmos.  In regions where no sentient beings currently exist, the force remains latent, akin to a magnetic field in a vacuum that manifests only in the presence of charged particles. This property aligns with other fundamental forces, which are considered invariant across time.

Some questions and answers will help clarify the three points above.

\vspace{2mm}

{\bf Q}. {\em Don't qualia depend on complex biological or neural systems, which are not present everywhere?}

{\bf A}. While the manifestation of qualia requires specific systems, the force itself can exist universally, waiting for suitable conditions to act upon, just as gravity exists everywhere but its effects depend on the presence of mass.

{\bf Q}. {\em Why is qualia force objective, since after all qualia are subjective experiences unique to individuals?}

{\bf A}. The subjectivity of qualia refers to their manifestation, not the functioning of the underlying force. The qualia force objectively generates experiences based on input, with subjectivity arising from the interaction with specific cognitive architectures.

{\bf Q}. {\em Since qualia are tied to temporal processes like perception and cognition, why is the force timeless?}

{\bf A}. Although qualia manifest within time-dependent living organisms, the underlying force responsible for their creation is not dependent on these organisms. This can be compared to electromagnetism, which exists timelessly as a fundamental force, while specific electromagnetic waves oscillate with time.

{\bf Q}. In the mathematical framework detailed in Sections 2.3 and 4.2, what represents absolute reality and what represents qualia force?

{\bf A}. The vectors in Hilbert space $\HH$ and the physical state functions $\x(t)$ represent absolute reality, whereas the inner product reflects the action of qualia force. For any observer $M\subset \HH$, the range of the perception operator $\per_M$ embodies the phenomenal aspect of absolute reality relative to $M$.

\vspace{2mm}

\subsection{Summary} The concept of a universal nature underlying life and worldly existence has been proposed in various forms throughout human history. Here, we highlight a few examples. Hindu philosophy Vedanta holds that {\em Brahman} is the ultimate, infinite, and eternal reality underlying all existence. It is formless, changeless, and the source of all that exists \cite{Sa}. Mahayana Buddhism theorizes the concept of a universal essence (known as the {\em Buddha nature}), which is neither rising nor ending, that underpins the existence and functioning of all sentient beings \cite{Su}. Plato's theory of {\em forms} posits that absolute reality consists of perfect, immutable archetypes (forms) of which the physical world is merely an imperfect reflection \cite{Pl}. Spinoza \cite{Sp} proposed a monistic view, asserting that a single substance underlies all existence, integrating both thought (mental reality) and extension (physical reality). Kant introduced the concept of {\em noumenon}, which exists independently of human perception and is distinct from the phenomena that we experience \cite{Ka}. More recent versions of such belief include Russel's neutral monist view \cite{Ru27} that there is a basic neutral (neither phenomenal nor physical) substance underlying both the phenomenal and the physical, and Bohm's {\em implicate order} which represents a deeper, underlying reality from which the observable universe (the {\em explicate order}) emerges \cite{Bo80}. In particular, Chalmers \cite{Ch95} made the following insightful comment: ``Consciousness might be taken as a fundamental feature of the world, alongside mass, charge, and space-time. We might take experience as fundamental, and work from there to see what theory can be constructed."

In summary, forces like gravity and the qualia force represent the noumenal aspects of absolute reality, with their manifestations (phenomena) perceived and interpreted uniquely by different observers. While qualia depends on particular physical systems (e.g., neurons, nerves, and brains), the force itself must operate according to objective laws, constituting a part of absolute reality. It complements the four forces by bridging the gap between mental and physical realms. All phenomena, whether physical or mental, must have causes grounded in the nature of absolute reality. The four forces and the qualia force provide the causal framework for both physical and mental phenomena. This approach avoids dualistic or supernatural explanations for consciousness, rooting it instead in the natural laws of absolute reality. From this perspective, the framework developed in this paper offers a promising logical foundation for developing a comprehensive worldview of ourselves and the physical world.

\section{An Excursion to Chan Buddhism}

Buddha nature, mentioned in Sect. 7.3, represents the intrinsic potential for enlightenment inherent in all living beings. Moments of enlightenment are often triggered by sudden sensory stimuli. Indeed, awakening to the presence of qualia is often regarded a hallmark of enlightenment \cite{Pu,Su91}. In his celebrated book {\em G\"{o}del, Escher, Bach: An Eternal Golden Braid}, Douglas Hofstadter made frequent references to Chan (Zen) stories and Gong'ans (Koans). To conclude this discussion, we follow his lead and draw on two Gong'ans to bring traditional wisdom to light. 

\vspace{2mm}

{\bf 1}. {\em Baizhang's Nose}. One day, Chan Master Mazu (709-788) and his student Baizhang (720-814) were strolling outside of the monastery. Some wigeons flew overhead. Mazu asked, ``What are those birds?" Baizhang replied, ``They are wigeons". As they continued walking, Mazu inquired again a moment later, ``Where are they now?" Baizhang responded, ``They flew away". Mazu suddenly pinched Baizhang's nose hard, and Baizhang cried out in pain, ``Pain! let go of my nose!" Mazu shouted, ``Didn't you just say they flew away?" Baizhang was awakened at that instant.

\vspace{1mm}

{\bf 2}. {\em The Sound of Bamboo}. When Baizhang became a Chan master, Zhixian (?- 898) came to study under him. Zhixian was naturally intelligent and well-versed in Buddhist doctrines. Whenever a question was posed to him, he could eloquently present 10 answers. After Baizhang passed away, Zhixian went to study with his senior Dharma brother, Weishan (771-853). One day, in an attempt to break Zhixian's cling to doctrines, Weishan questioned him, ``What was your original face before you were born?" Zhixian was completely baffled by this question. Finding no answer in the scriptures, he lamented, ``A painted rice cake cannot satisfy hunger." In despair, he burned all the scriptures and bid farewell to Weishan. Years later, one day while working in a field clearing weeds, he absentmindedly threw a piece of rock, which struck a bamboo stalk, producing a clear, resonant sound. At that moment, he suddenly attained enlightenment.

\vspace{2mm}

In the Chan tradition, all mental constructs, including concepts such as ``Buddha nature" and ``qualia force", dissolve at the moment of enlightenment.

\end{document}